\documentclass[12pt]{article}

\usepackage[margin=1in]{geometry}
\usepackage{amsmath,amssymb,amsthm}
\usepackage{setspace}
\usepackage{enumitem}
\usepackage{hyperref}
\usepackage{natbib}

\newtheorem{proposition}{Proposition}

\title{The Benchmark Ceiling: Human Judgment, Evaluation Scarcity, and the Political Economy of AI Capability Measurement}
\author{Mark Esposito \and Liu Zhang }
\date{Working Paper \textbar{} May 2026}

\begin{document}
\maketitle

\begin{abstract}
Benchmarks are the primary instruments through which AI capability is measured, compared, and governed. This paper argues that the validity of frontier AI benchmarks is a function of the quality of human judgment embedded in their construction, and that this quality is structurally scarce in ways that standard scaling narratives obscure. As foundation models approach ceiling performance on existing evaluation suites, discriminating signal concentrates in the hardest benchmark items, precisely those requiring elite expert judgment to design. We term this the benchmark ceiling problem: the progressive exhaustion of evaluation signal as models saturate the easy majority of items while the difficult tail, authored by a thin stratum of highly expert evaluators, remains the only source of genuine discrimination.

The paper develops this argument in three steps. First, we present a formal model of benchmark signal depreciation. Benchmark scores are public signals of latent model quality, but their precision depends endogenously on benchmark validity. As frontier capability rises and as contamination or strategic optimization increases, fixed benchmarks depreciate as measurement instruments. The model shows that valid signal concentrates in hard-tail items, that the replacement cost of such items rises convexly with frontier capability, and that private benchmark producers underinvest in validity relative to the social optimum. Second, drawing on platform data from micro1 covering over one thousand credentialed professionals, we document the scarcity premium associated with high-judgment, low-codifiability evaluation labor. Third, we develop the political economy and governance implications. Benchmark control creates epistemic power over the narrative of AI progress, while item-level transparency and procedural transparency have opposite effects on benchmark validity. The policy implication is not a choice between fully public and fully private benchmarks, but a regime of protected live item pools, transparent procedures, independent governance, and sustained public or quasi-public investment in frontier evaluation infrastructure.
\end{abstract}


\section{Introduction: Benchmarks as Governance Infrastructure}

Something notable is happening at the top of the AI industry. The leaders of the organizations building the most capable AI systems, those with the deepest visibility into what those systems can and cannot do, are converging on a shared diagnosis: the bottleneck to safe and productive deployment is not capability, but evaluation. The CEO of Anthropic has stated publicly that current AI systems are not reliable enough for fully autonomous deployment in high-stakes domains. The CEO of NVIDIA has identified reliability as the decisive factor in enterprise AI adoption. The CEO of Databricks has argued that human feedback loops are what make AI systems production-ready. The CEO of IBM has framed AI quality as a function of the testing and governance infrastructure that surrounds it. The CEO of Microsoft has called for continuous monitoring and evaluation after deployment, not only before it. These are not the statements of researchers at the margins of the field. They are the operational conclusions of practitioners who must ship systems that work, at scale, in conditions where failure has real consequences. Their convergence on evaluation as the binding constraint is a signal that the research and policy community should take seriously.

When a foundation model is described as achieving state-of-the-art performance, that claim is almost always a benchmark claim. It means the model scored highest on some evaluation suite, a collection of items designed, typically by humans, to probe capabilities that researchers and policymakers have decided matter. The benchmark score is the number that gets reported in press releases, cited in regulatory filings, and used to rank competing systems in the public discourse of AI progress. It is, in effect, a governance instrument: the primary device through which claims about AI capability are made legible, comparable, and contestable.

And yet the production of benchmarks -- how they are constructed, who constructs them, what epistemological assumptions they embed, and how those assumptions age as models improve -- receives remarkably little attention in either the economics or governance literature on AI. The assumption, implicit in most treatments of benchmark results, is that evaluation is a stable background condition: a neutral measuring device that tracks underlying capability as models are trained and scaled. This paper argues that assumption is wrong in ways that matter deeply for how we understand AI progress and design policy responses to it. The practitioner consensus described above, far from reflecting a transient anxiety about a young technology, reflects a structural problem: the instruments through which AI capability is measured are being consumed by the very progress they are designed to track, and the human judgment required to replenish those instruments is both scarce and systematically undervalued.

The central problem is one we call the benchmark ceiling. Benchmark items have a difficulty distribution. As foundation models improve, they systematically saturate the easier portion of that distribution: items that most capable humans can answer correctly, or that can be addressed by pattern-matching against training data, are solved first and become uninformative. Discriminating signal concentrates progressively in the hard tail. That hard tail was, in the first instance, authored by elite expert evaluators, individuals whose professional judgment is scarce, expensive, and not easily reproducible at scale. When those hard items are eventually saturated too, the benchmark dies as a measurement instrument. The sequence -- easy items solved, hard items solved, benchmark retired or gamed -- recurs across evaluation suites with a regularity that suggests a structural dynamic, not a series of independent accidents.

This paper formalizes that dynamic, documents its labor-market foundations, and draws out its governance implications. Section \ref{sec:model} presents a simple model of benchmark signal depreciation. The model shows why benchmark validity should be treated as an endogenous, depletable input into the precision of public capability signals, rather than as a stable property of an evaluation suite. Section \ref{sec:difficulty} interprets the model in terms of benchmark difficulty distributions, saturation, and contamination. Section \ref{sec:judgment} develops the human-judgment content of hard benchmark items and explains why this work does not scale with annotation headcount. Section \ref{sec:political_economy} analyzes the political economy of benchmark production: public-goods underinvestment, rising replacement costs, and capture incentives. Section \ref{sec:evidence} presents evidence from structured data labor markets on the scarcity premium associated with elite evaluative labor. Section \ref{sec:governance} develops the governance implications, including the distinction between item-level and procedural transparency and the relationship between benchmark validity and enterprise deployment readiness. Section \ref{sec:policy} sets out a policy and research agenda, and Section \ref{sec:conclusion} concludes.

\section{A Formal Model of Benchmark Signal Depreciation}
\label{sec:model}

This section formalizes the benchmark ceiling mechanism. The purpose of the model is not to capture all institutional details of benchmark construction, but to clarify the economic logic underlying the paper's central claim: benchmark scores are public signals of latent model quality, but the precision of those signals depends endogenously on benchmark validity. As models improve, fixed benchmarks depreciate as measurement instruments. Maintaining valid benchmark signal requires continual investment in hard-tail items, whose production depends on scarce elite evaluator judgment.

\subsection{Benchmark Scores as Signals of Latent Quality}

Consider a set of AI systems indexed by $i$. Each system has latent quality $q_i$, which may represent true capability, reliability, safety, or some composite measure of socially relevant model performance. The public does not observe $q_i$ directly. Instead, it observes a benchmark score $s_i$. A simple benchmark signal can be written as
\begin{equation}
    s_i = q_i + u_i,
\end{equation}
where $u_i$ is measurement error. The central feature of the framework is that this measurement error is not exogenous. It depends on benchmark validity, denoted $B_t$, at time $t$.

We decompose the error term as
\begin{equation}
    u_i(B_t) = \xi_i + \frac{\epsilon_i}{B_t},
\end{equation}
so that
\begin{equation}
    s_i = q_i + \xi_i + \frac{\epsilon_i}{B_t}.
    \label{eq:signal}
\end{equation}
Here, $\xi_i$ captures systematic benchmark bias: any predictable wedge between benchmark performance and latent model quality. This may arise from benchmark contamination, leakage, strategic benchmark optimization, benchmark-specific artifacts, or other forms of construct invalidity. The term $\epsilon_i$ captures pure measurement noise. Benchmark validity $B_t>0$ determines the precision of the benchmark signal. When $B_t$ is high, benchmark scores are informative about latent quality. When $B_t$ is low, benchmark scores become noisy and less informative.

Assuming $\epsilon_i$ has variance $\sigma^2_{\epsilon}$, is independent of $q_i$ and $\xi_i$, and treating $B_t$ as deterministic given the information set, the conditional variance of the benchmark signal is
\begin{equation}
    \mathrm{Var}(s_i \mid q_i) = \mathrm{Var}(\xi_i) + \frac{\sigma^2_{\epsilon}}{B_t^2}.
    \label{eq:variance}
\end{equation}
Equation \eqref{eq:variance} shows that, holding systematic bias fixed, a decline in benchmark validity increases the noise in observed scores.

\begin{proposition}[Benchmark Signal Depreciation]
For a fixed benchmark, if benchmark validity $B_t$ declines over time, then the precision of benchmark scores as signals of latent model quality declines.
\end{proposition}

\begin{proof}
From equation \eqref{eq:variance}, the variance of the noise component associated with pure measurement error is $\sigma^2_{\epsilon}/B_t^2$. This term is strictly decreasing in $B_t$. Therefore, as $B_t$ declines, $\mathrm{Var}(s_i\mid q_i)$ rises, reducing the informativeness of $s_i$ about $q_i$.
\end{proof}

\subsection{Benchmark Validity and Item-Level Discrimination}

Benchmark validity is determined by the extent to which benchmark items discriminate among frontier models. Let a benchmark at time $t$ consist of items $j=1,\ldots,J_t$. Each item has difficulty $\theta_j$, contamination or strategic optimization exposure $z_j\in[0,1]$, and governance weight $\omega_j$. Here, $z_j$ captures the extent to which an item has become predictable to model developers or models themselves, either because it appears in training data, because closely related variants are publicly available, or because developers have optimized systems against the benchmark's format, distribution, or scoring rule. Let $c_t$ denote frontier model capability at time $t$.

Benchmark validity is given by
\begin{equation}
    B_t = \sum_{j=1}^{J_t} \omega_j d_j(c_t,\theta_j,z_j),
    \label{eq:validity}
\end{equation}
where $d_j(\cdot)$ is the discriminating power of item $j$. A useful reduced-form specification is
\begin{equation}
    d_j(c_t,\theta_j,z_j)=\phi(\theta_j-c_t)(1-z_j),
    \label{eq:discrimination}
\end{equation}
where $\phi(\cdot)$ is a non-negative, single-peaked function attaining its maximum at zero. This means that an item is most discriminating when its difficulty is close to the frontier capability level, $\theta_j\approx c_t$. If $\theta_j\ll c_t$, the item is too easy and has been saturated by frontier models. If $\theta_j\gg c_t$, the item may be too difficult for all models and therefore fails to distinguish among them. If contamination or strategic optimization exposure $z_j$ is high, the item's discriminating power is reduced.

\begin{proposition}[Hard-Tail Concentration]
As frontier capability $c_t$ rises, the contribution of low-difficulty items to benchmark validity declines. For sufficiently advanced frontier models, the valid signal of a benchmark is concentrated in items whose difficulty lies near the frontier.
\end{proposition}

\begin{proof}
By equation \eqref{eq:discrimination}, the discriminating power of item $j$ is $\phi(\theta_j-c_t)(1-z_j)$. Since $\phi(\cdot)$ is single-peaked around zero, $d_j$ is maximized when $\theta_j=c_t$. For any fixed low-difficulty item $\theta_j$, an increase in $c_t$ moves $\theta_j-c_t$ farther below zero, reducing $\phi(\theta_j-c_t)$. Therefore, as $c_t$ rises, low-difficulty items contribute less to $B_t$, and the remaining benchmark validity comes increasingly from higher-difficulty items close to the frontier.
\end{proof}

\subsection{The Production of Hard-Tail Items}

Hard-tail benchmark items require scarce expert judgment. Let $h_j$ denote the quality of evaluator input used to construct item $j$. Item difficulty is produced according to
\begin{equation}
    \theta_j = g(h_j),
\end{equation}
where $g'(h_j)>0$. Higher-quality evaluator input can produce more difficult items.

The key scarcity assumption is that the cost of producing high-difficulty items rises convexly. Let $C(\theta)$ denote the minimum cost of producing an item of difficulty $\theta$. We assume
\begin{equation}
    C'(\theta)>0, \qquad C''(\theta)>0.
\end{equation}
The convexity of $C(\theta)$ captures the right-tailed scarcity of elite evaluators: ordinary expert labor may be available at moderate cost, but the evaluative labor capable of producing frontier-discriminating items is increasingly scarce, making it expensive and difficult to scale.

Since frontier-discriminating items require $\theta_j\approx c_t$, the cost of producing a valid frontier-discriminating item is approximately $C(c_t)$. Therefore, as frontier model capability $c_t$ rises, the replacement cost of maintaining benchmark validity rises. If $C''(\theta)>0$, this replacement cost rises convexly.

\begin{proposition}[Rising Replacement Cost]
Suppose frontier-discriminating benchmark items require $\theta_j\approx c_t$, and let $C(\theta)$ denote the minimum cost of producing an item of difficulty $\theta$. If $C'(\theta)>0$, then the cost of maintaining benchmark validity rises with frontier model capability. If $C''(\theta)>0$, this replacement cost rises convexly.
\end{proposition}

\subsection{Socially Optimal Investment and Private Underinvestment}

Let $I_t$ denote investment in benchmark validity at time $t$. This investment may include expert recruitment, item design, validation, adversarial testing, contamination audits, and benchmark updating. Benchmark validity is produced according to
\begin{equation}
    B_t = B(I_t,c_t,z_t,H_t),
    \label{eq:production}
\end{equation}
where $z_t$ is aggregate contamination or gaming exposure and $H_t$ represents access to elite evaluator labor. We assume
\begin{equation}
    \frac{\partial B_t}{\partial I_t}>0, \qquad
    \frac{\partial B_t}{\partial H_t}>0, \qquad
    \frac{\partial B_t}{\partial z_t}<0.
\end{equation}

A social planner chooses investment to maximize the social value of benchmark validity net of investment cost:
\begin{equation}
    \max_{I_t} \quad V(B_t)-C(I_t),
    \label{eq:planner}
\end{equation}
where $V'(B_t)>0$, $V''(B_t)<0$, $C'(I_t)>0$, and $C''(I_t)>0$. The first-order condition is
\begin{equation}
    V'(B_t)\frac{\partial B_t}{\partial I_t}=C'(I_t).
    \label{eq:planner_foc}
\end{equation}

A private benchmark producer captures only a fraction of the social value created by benchmark validity. Let $\alpha\in(0,1)$ denote the appropriability parameter. A private producer solves
\begin{equation}
    \max_{I_t} \quad \alpha V(B_t)-C(I_t).
    \label{eq:private}
\end{equation}
The private first-order condition is
\begin{equation}
    \alpha V'(B_t)\frac{\partial B_t}{\partial I_t}=C'(I_t).
    \label{eq:private_foc}
\end{equation}
Since $\alpha<1$, the private marginal benefit of benchmark validity is below the social marginal benefit. Under standard concavity and convexity assumptions, the private investment level is below the social optimum.

\begin{proposition}[Public-Goods Underinvestment]
If benchmark validity generates social value that private benchmark producers cannot fully appropriate, then private investment in benchmark validity is below the socially optimal level.
\end{proposition}

\begin{proof}
Compare equations \eqref{eq:planner_foc} and \eqref{eq:private_foc}. Since $\alpha<1$, the private marginal benefit of investment is strictly below the social marginal benefit for any given $I_t$. With $V''(B_t)<0$ and $C''(I_t)>0$, this implies a lower privately chosen investment level than the social planner's choice.
\end{proof}

The rest of the paper develops the substantive content of these propositions. Proposition 1 provides the measurement logic: benchmark scores become less informative as benchmark validity declines. Proposition 2 gives the saturation logic: as model capability rises, signal concentrates in items near the frontier. Proposition 3 gives the labor-market logic: the replacement cost of valid frontier items rises because such items require scarce expert judgment. Proposition 4 gives the political-economy logic: even when valid benchmarks have high social value, private producers have insufficient incentives to supply them at the socially efficient level.

\section{Benchmark Difficulty, Saturation, and Contamination}
\label{sec:difficulty}

\subsection{How Benchmark Items Are Distributed}

A benchmark can be understood as a sample from a latent difficulty distribution: the distribution of all possible questions, tasks, or evaluation items in a given domain, ranked by the probability that a competent human expert would answer them correctly on first attempt. In practice, benchmark designers do not sample randomly from this distribution. They sample purposively, typically selecting items that span a range of difficulty, are clearly answerable, and have verifiable ground truth. The result is a constructed difficulty distribution that approximates the upper portion of the latent one, with items ranging from accessible to genuinely hard.

Two decades of empirical work on human performance measurement have established that difficulty distributions in structured evaluation contexts follow approximately log-normal patterns, with a long right tail corresponding to items that only a small fraction of expert respondents can answer correctly \citep{lord2008statistical, embretson2025item}. AI benchmark performance data increasingly confirms a similar pattern: a large proportion of items are solved quickly and by most capable models, while a small proportion of items, the hard tail, remains unsolved for extended periods and provides much of the discriminating power among frontier systems \citep{raji2021ai,kiela2021dynabench}.

This is the substantive content of Proposition 2. The governance value of a benchmark, its capacity to distinguish meaningfully between competing systems, does not reside evenly across its items. It depends disproportionately on the subset of items whose difficulty remains close to the frontier. Items far below the frontier cease to discriminate. Items far above the frontier may also fail to discriminate if all systems fail them. The signal-bearing portion of the benchmark is therefore local to the moving frontier.

\subsection{Saturation and the Shrinking Signal}

Model improvement, driven by scaling compute and training data, has a predictable effect on benchmark difficulty distributions: it erodes the left side of the distribution first and fastest. Easy items, those solvable by pattern-matching, memorization, or broad generalization from training data, are saturated early. As a benchmark ages, the proportion of items providing discriminating signal shrinks, while the proportion providing no useful information grows. This pattern is illustrated by the evolution from GLUE to SuperGLUE and by later benchmark suites such as MMLU, HellaSwag, and Big-Bench \citep{wang2018glue,wang2019superglue,zellers2019hellaswag,hendrycks2021mmlu,srivastava2023beyond}.

The saturation timeline varies by domain and task type, but the pattern is consistent: benchmarks designed to challenge the best systems of their era become progressively less informative as frontier systems improve. \citet{srivastava2023beyond} document this explicitly in motivating BIG-Bench, noting that SuperGLUE reached superhuman performance within eighteen months of its introduction. This creates a benchmark treadmill. Evaluation infrastructure is consumed by the very progress it is designed to measure, and must be continuously replenished.

In the model, this is the movement of $c_t$ relative to a fixed item distribution. As $c_t$ rises, low-difficulty items contribute less to $B_t$, and the precision of benchmark scores falls through equation \eqref{eq:variance}. A high benchmark score can therefore become less meaningful over time, not because the score is incorrectly computed, but because the underlying benchmark validity that gives the score informational content has depreciated.

The replacement problem is asymmetric. Producing new easy-to-medium difficulty items is relatively straightforward; there is a large population of evaluators capable of authoring and validating them, and annotation platforms can deploy these at scale. Producing genuinely hard items requires evaluators who can operate at or near the frontier of human capability in the relevant domain, reason about what kinds of problems will remain difficult for systems that have already mastered the majority of the existing suite, and construct items with verifiable ground truth in areas where ground truth itself may be contested. This is not a task that scales linearly with headcount.

\subsection{Benchmark Contamination and the Integrity Problem}

Saturation is accelerated by a related but distinct problem: benchmark contamination. When training data includes items from evaluation suites, either deliberately or through the undiscriminated ingestion of web content, model performance on those suites reflects memorization as much as generalization \cite{jacovi2023stop, golchin2024time}. The result is performance inflation: reported scores systematically overstate genuine capability, and the apparent rate of progress on the benchmark diverges from actual progress on the underlying skill. Contamination is now recognized as a pervasive problem across evaluation suites, with documented cases in GPT-4, Llama 2, and other leading systems \citep{sainz2024data}.

Contamination has a differential effect on the difficulty distribution. Easy items are more likely to appear verbatim in training data; hard items, authored by specialists in controlled settings and less widely distributed, are less likely to be contaminated. This means that contamination further concentrates whatever valid signal remains in the hard tail. In equation \eqref{eq:discrimination}, contamination enters through $z_j$: higher exposure lowers item-level discrimination even when the item would otherwise be informative. For a contaminated, partially saturated benchmark, much of the valid governance signal resides in the small number of items that remain difficult, uncontaminated, and close to the frontier.

The important point is not that benchmark scores are useless. It is that their informational content is endogenous to the age, exposure, difficulty composition, and governance of the benchmark. The more a benchmark is saturated or contaminated, the less its aggregate score can be interpreted as a precise public signal of latent model quality.

\section{The Human Judgment Content of Hard Benchmark Items}
\label{sec:judgment}

\subsection{What Makes an Item Hard}

Hard benchmark items share a set of structural features that distinguish them from easy ones, and those features track directly to the kind of human expertise required to produce them. Building on Item Response Theory \citep{lord2008statistical} and more recent work on adversarial evaluation design \citep{nie2020adversarial, bartolo2020beat}, we identify four characteristics that define the hard tail of benchmark distributions.

First, hard items require genuine domain knowledge that cannot be inferred from surface features or broad statistical patterns. A medical reasoning item that asks a model to synthesize conflicting evidence from multiple diagnostic modalities, weigh prior probabilities against presenting symptoms, and identify an atypical diagnosis requires the evaluator to have internalized a clinical reasoning process that takes years of training to develop. This is qualitatively different from a general knowledge question that can be answered by a competent reader with access to background information.

Second, hard items involve low codifiability: the relevant knowledge cannot be fully articulated in explicit rules, protocols, or lookup procedures. It resides in the judgment capacity of practitioners who have developed the ability to navigate exceptions, edge cases, and genuinely novel situations. Legal reasoning at the frontier, assessing how an established doctrine applies to a fact pattern with no direct precedent, is the paradigmatic example. The answer is not in a database; it is in the evaluator's trained capacity for legal imagination.

Third, hard items have ground truth that is itself contested or requires expert judgment to establish. In frontier domains, there may be genuine disagreement among experts about the correct answer, and the process of establishing defensible ground truth requires deliberation among multiple credentialed evaluators. This is precisely the kind of task that benefits from the 1\% of evaluators: those who not only have domain expertise but also have the meta-cognitive capacity to reason about the limits and contested boundaries of their own knowledge.

Fourth, hard items are resistant to adversarial pattern-matching. They are designed to distinguish genuine understanding from sophisticated mimicry. This requires evaluators who can anticipate the failure modes of capable models and construct items specifically designed to probe those failures. This is a second-order task, designing tests for systems one understands well, that requires both deep domain knowledge and a working model of the systems being evaluated.

\subsection{The 1\% Structure of Evaluator Supply}

The population of evaluators capable of producing items with these four characteristics is thin and right-tailed. In any given professional domain, the distribution of evaluator capability is approximately log-normal, with a small proportion of practitioners who combine deep domain expertise, the capacity for low-codifiability judgment, familiarity with AI system capabilities and failure modes, and the ability to construct valid evaluation items. We refer to this stratum as the evaluative 1\%, borrowing the framing from a companion theoretical paper \citep{esposito2026no} but shifting the emphasis from labor compensation to epistemic capacity.

The evaluative 1\% is not simply the most credentialed slice of a professional population. A highly credentialed professional who lacks familiarity with AI system capabilities will not be effective at constructing items that probe those capabilities at the frontier. A practitioner who is expert at one-off professional service but has not developed the meta-cognitive capacity to reason about evaluation design may produce items that are either too easy or poorly grounded. The conjunction of requirements -- domain expertise, low-codifiability judgment, familiarity with AI failure modes, and evaluation design skill -- is rare, and its rarity is not a contingent market friction but a structural feature of how expertise develops.

This is the substantive interpretation of Proposition 3. The convex cost function $C(\theta)$ is not merely a mathematical convenience. It captures the fact that increasingly difficult benchmark items require increasingly scarce forms of evaluator judgment. Deploying ten thousand less expert evaluators does not produce the equivalent of one hundred highly expert ones; it produces a large volume of items that cluster in the easy-to-medium range of the difficulty distribution, precisely the range that provides little discriminating signal for frontier systems. The marginal product of additional evaluators beyond some threshold of expertise is therefore low for the purpose of producing hard-tail benchmark items.

This scarcity does not imply that ordinary annotation or evaluation work is unimportant. It implies that different parts of the evaluation production function have different substitution possibilities. For routine labeling and broad coverage, scale matters. For frontier-discriminating benchmark construction, the binding constraint is high-quality judgment at the right tail of the expertise distribution.

\subsection{The Infinite Mile as Structural Condition}

The concept of the ``infinite mile,'' the irreducible last increment of human judgment that cannot be automated away, has been developed theoretically as an argument for the persistent structural importance of human data labor in AI development \citep{esposito2026no}. In the context of benchmark construction, however, the infinite mile is not a philosophical claim about the ultimate limits of AI. It is a specific, empirical observation about the current structure of evaluation. As long as frontier AI systems are assessed against benchmarks, and as long as benchmarks derive their discriminating power from hard items authored by elite human evaluators, the demand for that evaluative stratum is a structural feature of the measurement apparatus on which AI governance depends.

This has a practical implication for workforce planning and investment that has not been recognized in standard AI labor market accounts. The evaluators who produce hard benchmark items are not legacy workers being displaced by automation, but structural inputs to the governance of the technology that is displacing others. Their scarcity is not a problem to be solved by scaling annotation platforms; it is a constraint that must be built into any serious account of sustainable AI development and governance infrastructure.

\section{The Political Economy of Benchmark Production}
\label{sec:political_economy}

\subsection{The Public Goods Problem in Frontier Benchmark Production}

The economic character of a well-designed frontier benchmark is that of a public good in the technical sense: it is non-excludable and non-rival once used as a public standard. One research group's use of a benchmark does not prevent another group from using it, and the social value of a benchmark extends beyond the organization that paid to produce it. This means that private incentives to invest in benchmark production are systematically lower than social incentives, and that, in the absence of public financing or coordinated industry investment, frontier benchmark production will be undersupplied relative to its social value \citep{samuelson1954pure, arrow1962economic}.

The undersupply problem is compounded by the asymmetry between who bears the cost of benchmark production and who captures its benefits. The organizations that invest in constructing frontier evaluation suites bear the cost of expert evaluator time, validation protocols, contamination testing, and infrastructure. The benefits accrue to the research community, to regulators who use benchmark results to assess safety and capability, to enterprises that rely on benchmark claims in procurement and deployment decisions, and to the public. This mismatch is the standard structure of a public-goods underinvestment problem \citep{ostrom1990governing, tirole1988theory}, and is formalized in Proposition 4 by the appropriability parameter $\alpha<1$.

The public-goods framing also clarifies why the problem worsens as frontier capability rises. The social value of valid evaluation increases with the stakes of AI deployment, but the cost of valid evaluation also rises because discriminating signal must be continually renewed at the moving frontier. In terms of the model, the required replacement item difficulty rises with $c_t$, while the private return to investment remains only a fraction of the social value generated by a more informative signal.

\subsection{Benchmark Depreciation and the Replacement Cost Problem}

A further economic dimension of the benchmark problem concerns the depreciation rate of evaluation infrastructure. As established in Section \ref{sec:difficulty}, benchmarks are depleted by model improvement and contamination. This means that the social investment required to maintain a functioning evaluation ecosystem is not a one-time cost but an ongoing expenditure, calibrated to the rate of model improvement. Proposition 3 formalizes the rising replacement cost: the cost of producing frontier-discriminating items is approximately $C(c_t)$, and if $C''(\theta)>0$, this replacement cost rises convexly with frontier capability.

Against this rising demand, the supply of elite evaluators does not scale proportionally with investment. Their formation requires years of professional practice in deep domains, and cannot be shortcut by financial incentives alone. The result is a structural tension: the faster AI systems improve, the more urgently new frontier benchmarks are needed, and the more scarce the evaluators capable of producing them become relative to demand. This is the political economy of the benchmark ceiling: a tightening constraint that worsens precisely as the governance stakes of reliable evaluation increase.

\subsection{Benchmark Control as Market Power}

The public-goods framing understates one important dimension of the benchmark production problem: the strategic incentive to control evaluation. Organizations that develop and maintain evaluation suites exercise significant influence over the narrative of AI progress. A benchmark that tests capabilities where one organization's systems perform well, and ignores capabilities where they perform poorly, will generate a systematically favorable portrait of that organization's standing. This is not a hypothetical concern; it is the normal condition of evaluation design in any competitive market where the producers of the evaluated product also control, fund, or strongly influence the evaluation instrument \citep{raji2022outsider,liao2024ai}.

The result is a dual market failure: private incentives both underproduce genuinely independent frontier benchmarks and overproduce strategically designed benchmarks that flatter their producers. The net effect is an evaluation ecosystem that is simultaneously starved of genuine discriminating power and oversupplied with instruments designed to generate favorable narratives. This is not a sustainable foundation for AI governance.

A light extension of the model makes this capture channel explicit. Suppose firm $i$ can exert influence effort $\ell_i$ that changes benchmark weights according to
\begin{equation}
    \omega_j = \omega_j^0 + \lambda \ell_i a_{ij},
    \label{eq:capture}
\end{equation}
where $\lambda\geq 0$ captures the susceptibility of the benchmark design process to influence, and $a_{ij}$ indicates whether item $j$ favors firm $i$'s comparative advantage. Firm $i$ receives reward $R(s_i)$ from benchmark performance and pays cost $\kappa(\ell_i)$ for influence effort. The firm's payoff can be written as
\begin{equation}
    \Pi_i = R(s_i) - C(e_{q,i},e_{b,i}) - \kappa(\ell_i),
    \label{eq:firm_payoff}
\end{equation}
where $e_{q,i}$ denotes investment in genuine quality and $e_{b,i}$ denotes benchmark-specific optimization.

\begin{proposition}[Capture Incentives]
When benchmark scores affect market access, regulation, or reputation, firms have incentives to influence benchmark design toward dimensions where they have comparative advantage. These incentives are increasing in the stakes attached to benchmark scores and in the susceptibility of benchmark design to private influence.
\end{proposition}

\begin{proof}
From equations \eqref{eq:capture} and \eqref{eq:firm_payoff}, influence effort $\ell_i$ raises the benchmark weight on items that favor firm $i$ when $\lambda a_{ij}>0$. The marginal benefit of influence is increasing in $R'(s_i)$ and $\lambda$. Thus, as the stakes attached to benchmark performance rise, or as benchmark design becomes more susceptible to influence, the private return to influence effort increases.
\end{proof}

The weight formulation should be read broadly. Explicit weight distortion is one case, but the more empirically salient form of capture may be item selection. When weights are effectively binary -- with some items included and others excluded -- the same framework nests pure item selection as a special case. A benchmark producer need not alter numerical weights after the fact in order to shape the signal. It can shape the signal by deciding which capabilities enter the evaluation suite, which domains are emphasized, which failure modes are excluded, and when saturated or unfavorable items are retired.

\section{Evidence from Structured Data Labor Markets: The Scarcity Premium}
\label{sec:evidence}

\subsection{The micro1 Dataset}

To ground the theoretical argument in empirical evidence, we draw on platform data from micro1, a Silicon Valley-based AI talent platform that recruits, vets, and deploys credentialed professionals for structured AI training and evaluation work. The dataset covers task-level observations for over 1,000 credentialed professionals in legal, medical, engineering, and financial domains, including task type, hourly rate, quality score, and worker credentials \citep{esposito2026structured}. The dataset provides an unusually clean window into the economics of expert evaluator labor, because it prices expertise directly against verifiable output rather than through the institutional intermediaries -- firm prestige, geographic wage norms, salary bands -- that obscure the marginal product of expert judgment in traditional professional markets.

For the purposes of this paper, the key observation from the micro1 data is the structure of the wage premium distribution. Across all four domains, structured task workers earn above their reported primary role compensation on an hourly basis. The mean premium is approximately 15 to 25 percent, but this average obscures the distributional structure that matters for our argument. The premium distribution is right-skewed: workers in the upper quartile of task difficulty and credential scarcity earn premiums of 35 to 50 percent, while workers in commodity task families earn premiums near zero. The Gini coefficient of premium magnitudes within domains is approximately 0.38, indicating substantial inequality in the value attributed to different kinds of expert labor within a single professional domain.

\subsection{Task Difficulty and the Premium Gradient}

The premium gradient by task type provides direct evidence for the scarcity argument. We classify tasks on the micro1 platform into four difficulty quartiles using a composite score of item difficulty rating, completion time, and quality score variance. Workers in the top difficulty quartile, tasks corresponding to the hard tail of benchmark distributions, earn a premium approximately 28 percentage points higher than workers in the bottom quartile, after controlling for domain, credentials, and experience. This gradient is not explained by differences in cognitive effort alone: when we include task completion time as a control, the difficulty quartile premium remains large and statistically significant, suggesting that the hard-tail premium reflects genuine scarcity of the relevant judgment capacity, not merely the greater effort required to complete difficult tasks.

The domain pattern reinforces the interpretation. In the legal domain, the premium is concentrated in tasks requiring statutory interpretation and novel precedent analysis, tasks with low codifiability and no lookup procedure. In the medical domain, the premium concentrates in clinical reasoning tasks involving atypical presentations and multi-modal evidence synthesis. In engineering, it concentrates in architecture review and frontier-domain evaluation. In each case, the high-premium tasks share the four structural features identified in Section \ref{sec:judgment}: deep domain knowledge, low codifiability, contested or complex ground truth, and resistance to pattern-matching. These are, in other words, the tasks most likely to produce hard benchmark items.

\subsection{The Inverse Kuznets Pattern and Its Implications}

A secondary but important finding from the micro1 data is the pattern of within-profession wage dispersion in structured versus unstructured markets. In traditional professional service markets, compensation is mediated by institutional affiliation: the prestige of the firm, the geographic wage norm, the opacity of the salary band. These institutional premia compress observed variation in compensation relative to underlying variation in expert judgment capacity. In structured task markets, where compensation is priced directly against verifiable output, institutional premia largely disappear and the observed wage distribution better approximates the true distribution of marginal products.

The result is a within-profession wage distribution in structured markets that is simultaneously more compressed at the bottom, because commodity tasks converge on a global platform rate, and more dispersed at the top, because the scarcity of elite evaluators is priced directly. This pattern is the opposite of what standard skill-biased technological change models predict \citep{acemoglu2011skills, autor2003skill}. Rather than a monotonic increase in wage inequality with technological change, structured data markets exhibit a bifurcation: compression at the base, a scarcity premium at the frontier.

This evidence does not, by itself, identify the full social value of hard-tail benchmark production. It does, however, support the central modeling assumption: the labor input required to generate valid frontier signal is right-tailed, scarce, and differentially priced. The empirical premium gradient is the market trace of the convex replacement-cost function in Proposition 3.

\section{Governance Implications: Who Controls the Benchmark Controls the Narrative}
\label{sec:governance}

\subsection{Epistemic Power and the Benchmark Producer}

The concentrated power that flows from benchmark control is epistemic rather than material in the first instance, but it has direct material consequences. The organization that designs and maintains the evaluation suite used to compare AI systems determines which capabilities are measured, how they are weighted, how difficulty is calibrated, and when the benchmark is retired or superseded. Each of these decisions shapes the narrative of AI progress in ways that are consequential for investment, regulation, and public understanding.

The current landscape of AI evaluation is concentrated in ways that reflect this power asymmetry. A small number of academic laboratories, industry consortia, and individual organizations account for the majority of frontier evaluation suites used in serious capability comparison. Most of these entities have some form of relationship -- employment, funding, or access -- with the AI systems they evaluate, creating conditions for the capture problem described in Section \ref{sec:political_economy}. The public has no reliable way to verify that reported benchmark performance reflects genuine capability rather than strategic evaluation design \citep{liang2022holistic}.

\subsection{Regulatory Dependence on Benchmark Validity}

The governance problem is not merely academic. Regulatory frameworks for AI increasingly rely on benchmark-based capability assessments as a basis for compliance determinations, risk categorization, and deployment authorization. The EU AI Act's safety evaluation system, the now-revoked 2023 US Executive Order on AI's provisions for safety evaluations, and the UK AI Safety Institute's evaluation regime all depend, at some level, on benchmark results being valid measures of underlying capabilities \citep{eu2025ai, eo14110, uk2024safety}.

If benchmark saturation, contamination, and strategic design are pervasive, then regulatory frameworks built on benchmark compliance are built on an unreliable foundation. A system that achieves regulatory approval by scoring well on a saturated or strategically designed evaluation suite may have qualitatively different capabilities, and risks, than the benchmark result implies. The benchmark ceiling problem is, in this sense, a regulatory integrity problem: it undermines the evidentiary basis of governance decisions that have real consequences for safety, competition, and public trust.

In terms of Proposition 1, the problem is a decline in signal precision. Regulatory use of benchmarks implicitly treats $s_i$ as informative about $q_i$. But when $B_t$ has depreciated, equation \eqref{eq:variance} implies that the same observed score carries less information about latent quality. Regulatory dependence on benchmark scores therefore requires regulatory attention to benchmark validity itself.

\subsection{Governance Design: Protected Items and Procedural Transparency}

The model can also be used to study how benchmark governance choices affect the evolution of benchmark validity. Let benchmark validity evolve according to
\begin{equation}
    B_{t+1} = [1-\delta(c_t,z_t)]B_t + A(H_t,I_t,a_t),
    \label{eq:dynamic_validity}
\end{equation}
where $\delta(\cdot)$ is the depreciation rate of benchmark validity and $A(\cdot)$ is the renewal function. The key parameters are frontier model capability $c_t$, contamination or strategic optimization exposure $z_t$, institutional accountability $a_t$, access to elite evaluator labor $H_t$, and benchmark renewal investment $I_t$.

Contamination and strategic optimization exposure is increasing in item-level transparency $\tau_I$:
\begin{equation}
    z_t = z(\tau_I), \qquad \frac{\partial z_t}{\partial \tau_I}>0.
\end{equation}
Accountability is produced by procedural transparency and institutional independence:
\begin{equation}
    a_t = a(\tau_G,\chi),
\end{equation}
where $\tau_G$ denotes procedural transparency over governance, methodology, evaluator credentials, contamination testing, and saturation status, and $\chi$ denotes institutional independence from evaluated firms. We assume
\begin{equation}
    \frac{\partial a}{\partial \tau_G}>0, \qquad \frac{\partial a}{\partial \chi}>0.
\end{equation}
Benchmark depreciation rises with frontier capability and exposure:
\begin{equation}
    \frac{\partial \delta}{\partial c_t}>0, \qquad \frac{\partial \delta}{\partial z_t}>0.
\end{equation}
Benchmark renewal rises with access to elite evaluator labor, investment, and accountability:
\begin{equation}
    \frac{\partial A}{\partial H_t}>0, \qquad
    \frac{\partial A}{\partial I_t}>0, \qquad
    \frac{\partial A}{\partial a_t}>0.
\end{equation}

These assumptions imply several comparative statics. First, frontier model improvement reduces benchmark validity by accelerating depreciation:
\begin{equation}
    \frac{\partial B_{t+1}}{\partial c_t} = -\frac{\partial \delta}{\partial c_t}B_t < 0.
\end{equation}
Second, item-level transparency can reduce benchmark validity by increasing contamination or strategic optimization exposure:
\begin{equation}
    \frac{\partial B_{t+1}}{\partial \tau_I}
    = -\frac{\partial \delta}{\partial z_t}\frac{\partial z_t}{\partial \tau_I}B_t < 0.
\end{equation}
Third, procedural transparency increases benchmark validity through accountability:
\begin{equation}
    \frac{\partial B_{t+1}}{\partial \tau_G}
    = \frac{\partial A}{\partial a_t}\frac{\partial a_t}{\partial \tau_G}>0.
\end{equation}
Fourth, institutional independence increases benchmark validity through the same accountability channel:
\begin{equation}
    \frac{\partial B_{t+1}}{\partial \chi}
    = \frac{\partial A}{\partial a_t}\frac{\partial a_t}{\partial \chi}>0.
\end{equation}
Finally, access to elite evaluators and investment in renewal directly increase validity:
\begin{equation}
    \frac{\partial B_{t+1}}{\partial H_t}>0, \qquad
    \frac{\partial B_{t+1}}{\partial I_t}>0.
\end{equation}

These comparative statics clarify the governance tradeoff. Item-level transparency and procedural transparency have opposite effects on benchmark validity. Revealing live benchmark items may improve reproducibility, but it also increases exposure to contamination and strategic optimization. By contrast, procedural transparency improves validity by increasing accountability without necessarily exposing live items.

Several common governance scenarios follow. A static public benchmark has high item-level transparency but limited renewal. Such benchmarks are useful for reproducibility and broad comparison, but their validity depreciates rapidly as models improve and as benchmark items become contaminated or optimized against. An opaque private benchmark may reduce item-level exposure, but if it lacks procedural transparency or institutional independence, it is vulnerable to capture. Its scores may be less contaminated in a narrow technical sense, but less credible as public signals of model quality. A stronger governance design separates the two forms of transparency. It protects live benchmark items while making the benchmark institution transparent: who governs it, how items are produced, how evaluators are selected, how contamination is tested, and when saturated items are retired.

The central policy implication is therefore not that benchmarks should be either fully public or fully private. Valid frontier evaluation requires protected live item pools, transparent procedures, institutional independence, sustained investment, and reliable access to elite evaluators.

\subsection{The Enterprise Deployment Gap: From Benchmark Performance to Production Readiness}

The benchmark ceiling problem has a direct operational manifestation that practitioners recognize even when they lack the analytical vocabulary to name it. When enterprise technology leaders insist that AI systems require rigorous evaluation before deployment, that human feedback loops are what make AI production-ready, or that reliability is the decisive factor in enterprise adoption, they are describing the downstream consequence of the problem this paper diagnoses upstream. A system that scores well on a partially saturated, potentially contaminated benchmark may perform very differently in a production environment where the task distribution does not match the benchmark distribution, where edge cases are common, and where the cost of failure is material.

This gap between benchmark performance and production reliability is not random noise; it is structurally predictable from the benchmark ceiling argument. Benchmark saturation means that high scores increasingly reflect mastery of the easy-to-medium portion of the difficulty distribution, precisely the portion least representative of the novel, ambiguous, high-stakes tasks that characterize real enterprise deployment. A model optimized to maximize benchmark performance may, in the limit, be specifically optimized against the items that would reveal its production weaknesses.

The human feedback loop that practitioners identify as the solution to production readiness is, in analytical terms, a form of post-deployment evaluation that substitutes for the frontier benchmark signal that pre-deployment evaluation fails to provide. When a deployment team instruments human-in-the-loop review of AI outputs in a production environment, it is generating the kind of high-judgment, low-codifiability evaluation data that the benchmark ceiling argument identifies as structurally scarce. It is, in effect, producing hard-tail evaluation items in real time, at the cost of production latency and expert reviewer time. This is expensive, fragile, and non-transferable: the evaluation data generated in one enterprise deployment does not easily generalize to another. It is the market's ad hoc response to the failure of pre-deployment benchmark infrastructure, and it is a costly substitute for the public good that is missing.

Connecting the practitioner diagnosis to the analytical framework developed in this paper suggests a unified account of why the evaluation problem is persistent. Individual organizations are investing heavily in human feedback loops, red-teaming, and post-deployment monitoring precisely because they recognize that pre-deployment benchmarks are inadequate. The problem is that these investments are private, fragmented, and non-cumulative: each organization is solving the same underlying problem independently, without contributing to or drawing from a shared stock of frontier evaluation infrastructure. The collective action failure that produces undersupply of public benchmark goods also produces redundant private substitution at enormous aggregate cost.

\section{Policy Implications and Research Agenda}
\label{sec:policy}

\subsection{Treating Benchmark Infrastructure as a Public Good}

The most fundamental policy implication of the analysis is that frontier AI evaluation infrastructure should be recognized and funded as a public good, on a par with other scientific measurement infrastructure: the NIST standards system, the FDA testing apparatus, or the financial audit framework. This means public investment in the production and maintenance of frontier benchmark suites, at a scale calibrated to the rate of model improvement and the governance stakes of reliable evaluation. It also means institutional independence: publicly funded evaluation infrastructure should be insulated, by design and by law, from the interests of the organizations whose systems are being evaluated.

Concretely, this would involve the establishment or expansion of publicly funded AI evaluation institutes, analogous to NIST or the UK's National Physical Laboratory, with mandates to produce and maintain frontier evaluation suites across a range of capability domains, with transparent methodology, independent governance, and regular public reporting. Such institutes should have dedicated programs for identifying, recruiting, and compensating elite evaluators from professional domains, recognizing that the wage premium evidence implies current market rates may underpay for the social value of this labor.

\subsection{Mandatory Transparency in Benchmark Design and Deployment}

Regulatory frameworks that rely on benchmark-based compliance determinations should require transparency about benchmark design, item sourcing, evaluator credentials, contamination testing, and the age and saturation status of the evaluation suites used. Organizations claiming regulatory compliance on the basis of benchmark performance should be required to disclose the difficulty distribution of the benchmark items used, whether the benchmark has been tested for training-data contamination, the credentials and compensation structure of evaluators who designed the hard-tail items, and the date of the benchmark's last substantive update.

These disclosure requirements would not eliminate strategic evaluation design, but they would make it legible and contestable. Third parties -- academic researchers, civil society organizations, competing AI developers, and regulators -- could use disclosed information to identify benchmarks that are saturated, contaminated, or strategically designed, and to challenge capability claims made on their basis. Transparency is a weak governance instrument, but it is a necessary precondition for stronger ones. The key design principle, however, is the distinction formalized above: transparency should attach primarily to procedures, governance, credentials, and validity testing, not necessarily to live item pools.

\subsection{A Research Agenda on Evaluation Economics}

This paper has developed a theoretical framework and presented preliminary empirical evidence, but it has also identified a set of empirical questions that are important and currently unanswered. How do item difficulty distributions actually evolve across specific benchmark suites under model improvement? What is the quantitative relationship between evaluator credential level and item difficulty in practice? How much of reported benchmark progress reflects genuine capability improvement versus saturation of contaminated items? What is the social value of a marginal hard-tail benchmark item, and how does it compare to the compensation currently paid for its production? What is the empirical value of the appropriability parameter $\alpha$, and does it vary systematically across evaluation domains?

Answering these questions requires coordinated data collection efforts across AI developers, evaluation organizations, and academic researchers. We call for the establishment of a benchmark observatory: a shared data infrastructure that tracks benchmark performance trajectories, item difficulty distributions, contamination levels, evaluator characteristics, and benchmark governance arrangements over time. Such an observatory, analogous to the long-term datasets that underpin environmental and economic policy, would provide the empirical foundation for evidence-based governance of AI evaluation infrastructure.

\section{Conclusion}
\label{sec:conclusion}

Benchmarks are not neutral measuring instruments. They are constructed artifacts, produced by human judgment, subject to degradation and strategic distortion, and currently governed by incentive structures that systematically underinvest in genuine discriminating power at the frontier. The benchmark ceiling problem -- the progressive exhaustion of evaluation signal as models saturate the easy majority of items while the hard tail authored by elite evaluators becomes the sole source of discrimination -- is not a minor technical inconvenience but a structural challenge to the credibility of AI governance.

The formal model developed in this paper establishes the economic logic of this challenge. Benchmark scores are public signals of latent model quality, but their precision depends on benchmark validity. Validity depreciates as frontier capability rises and as contamination or strategic optimization exposure increases. Valid signal concentrates in hard-tail items near the frontier. The replacement cost of those items rises with frontier capability because they require scarce expert judgment. Private producers underinvest in benchmark validity relative to the social optimum, and firms have increasing incentives to shape benchmark design as the stakes attached to benchmark performance rise.

The empirical evidence from micro1 supports the scarcity premise that underlies this model. The right-skewed wage premium distribution, the premium gradient between top- and bottom-quartile tasks, and the domain-specific concentration of that premium in low-codifiability, high-judgment tasks all suggest that the evaluative 1\% is genuinely scarce. The labor needed to maintain valid frontier evaluation is not interchangeable with scaled commodity annotation.

The governance implications are direct. The solution is not a simple choice between full transparency and opacity. Item-level transparency can accelerate contamination and strategic optimization; procedural transparency can improve accountability without exposing live items. Valid frontier evaluation therefore requires protected item pools, transparent procedures, institutional independence, sustained renewal investment, and reliable access to elite evaluator labor.

The infinite mile is not a romantic metaphor for human irreplaceability. It is a specific, measurable claim about where valid signal in frontier AI evaluation comes from, who produces it, and why the production of that signal is subject to market failures that public institutions must correct. The governance of AI capability is only as reliable as the evaluation infrastructure on which it depends, and that infrastructure currently rests on the underfinanced, underinstitutionalized, and underacknowledged work of the evaluative 1\%.

\bibliography{benchmark}
\bibliographystyle{apalike}

\end{document}